%% file: MaxEnt.tex
\documentclass[runningheads,a4paper]{llncs}

\usepackage{amssymb}
\usepackage{graphicx}

\newcommand{\keywords}[1]{\par\addvspace\baselineskip
\noindent\keywordname\enspace\ignorespaces#1}

\usepackage{enumitem}

\usepackage[
  colorlinks,
  linkcolor = blue,
  citecolor = blue,
  urlcolor = blue]{hyperref}
\usepackage{mathtools}
\usepackage{thmtools}
\usepackage{thm-restate}

\newcommand{\N}{\mathbb{N}}
\newcommand{\C}{\mathbb{C}}

\newcommand{\M}{\mathrm{M}}

\newcommand{\ket}[1]{|#1\rangle}
\newcommand{\bra}[1]{\langle#1|}

\newcommand{\ketbras}[1]{|#1\rangle\langle#1|}
\newcommand{\ketbra}[2]{|#1\rangle\langle#2|}

\DeclarePairedDelimiter{\set}{\lbrace}{\rbrace}

\newcommand{\tp}{^{\textrm{T}}}
\newcommand{\ct}{^{\dagger}}
\newcommand{\x}{\otimes}

\newcommand{\id}{\ensuremath{\mathop{\rm Id}\nolimits}}
\newcommand{\ie}{\textit{i.e.}}

\newcommand{\bip} [2]{\C^{#1}\x\C^{#2}}    
\newcommand{\bips}[1]{\bip{#1}{#1}}        

\newcommand{\EQ}{\textrm{EQ}_n}
\newcommand{\EQX}{\textrm{EQ}(X)}
\newcommand{\EQP}[1]{\textrm{EQ}(#1)}

\renewcommand\vec{\operatorname{vec}}
\DeclareMathOperator{\tr}{Tr}

\DeclareMathOperator{\supp}{supp}

\DeclareMathOperator{\Pos}{Pos}

\newcommand{\Thm}[1]{\hyperref[thm:#1]{Theorem~\ref*{thm:#1}}}
\newcommand{\Lem}[1]{\hyperref[lem:#1]{Lemma~\ref*{lem:#1}}}
\newcommand{\Cor}[1]{\hyperref[cor:#1]{Corollary~\ref*{cor:#1}}}
\newcommand{\Def}[1]{\hyperref[def:#1]{Definition~\ref*{def:#1}}}
\newcommand{\Obs}[1]{\hyperref[obs:#1]{Observation~\ref*{obs:#1}}}

\newcommand{\Sect}[1]{\hyperref[sec:#1]{Section~\ref*{sec:#1}}}
\newcommand{\Apx}[1]{\hyperref[apx:#1]{Appendix~\ref*{apx:#1}}}

\newcommand{\Fig}[1]{\hyperref[fig:#1]{Figure~\ref*{fig:#1}}}
\newcommand{\Tab}[1]{\hyperref[tab:#1]{Table~\ref*{tab:#1}}}
\newcommand{\EqRef}[1]{\hyperref[eq:#1]{(\ref*{eq:#1})}}
\newcommand{\Eq}[1]{Equation~\hyperref[eq:#1]{(\ref*{eq:#1})}}

\begin{document}

\mainmatter  

\title{Maximally entangled states in pseudo-telepathy games}
\titlerunning{Maximally entangled states in pseudo-telepathy games}

\author{Laura Man\v{c}inska\inst{1}}
\authorrunning{Laura Man\v{c}inska}

\institute{Centre for Quantum Technologies, National University of Singapore
}

\toctitle{Maximally entangled states in pseudo-telepathy games}
\tocauthor{Laura Man\v{c}inska}
\maketitle

\begin{abstract}
A pseudo-telepathy game is a nonlocal game which
can be won with probability one using some finite-dimensional quantum strategy but not using a classical one. Our central question is whether there exist two-party pseudo-telepathy games which cannot be won with probability one using a maximally entangled state. Towards answering this question, we develop conditions under which maximally entangled states suffice. In particular, we show that maximally entangled states suffice for weak projection games which we introduce as a relaxation of projection games. Our results also imply that any pseudo-telepathy weak projection game yields a device-independent certification of a maximally entangled state. In particular, by establishing connections to the setting of communication complexity, we exhibit a class of games $G_n$ for testing  maximally entangled states of local dimension $\Omega(n)$. We leave the robustness of these self-tests as an open question.


\end{abstract}

\section{Introduction}
Entanglement is a central feature of quantum information processing (see \cite{Gruska} or \cite{Nielsen} for a general introduction). In many cases it can be used to perform nonlocal tasks that would otherwise be impossible or very inefficient. Therefore, entanglement is a resource and one is interested in means of measuring the entanglement content of a quantum state. Many such entanglement measures have been introduced and studied (see \cite{Plenio05} for a survey). The usual approach is to define entanglement as the resource that cannot be increased using local quantum operations and classical communication (LOCC). Here, we will only be concerned with the two-party scenarios. In such a case, any desired shared state can be obtained via LOCC from the so-called \emph{maximally entangled state} 
\[ 
  \ket{\Psi_d} := \frac{1}{\sqrt{d}}\sum_{i=1}^d \ket{ii}.
\]
Therefore, according to any entanglement measure,
the state $\ket{\Psi_d}$ possesses the highest entanglement content among all states in $\bips{d}$. It would perhaps be natural to expect that maximally entangled states are the most useful ones for accomplishing nonlocal tasks. It turns out that this intuition fails and there are known examples where less entangled states allow for better performance \cite{Acin02,Acin05,Zohren08}. Moreover, in some cases, maximally entangled states are shown to be suboptimal even if we do not restrict the dimension $d$ of the maximally entangled states  \cite{Junge11,Vidick10,Liang11,Regev12}. Most of these examples are given in terms of Bell inequality violations and when stated in terms of nonlocal games they yield games at which entangled players cannot succeed with probability one. Therefore, it could be possible that maximally entangled states are sufficient to achieve perfect performance, whenever this can be done using some entangled state. 


In a one-round two-party nonlocal game $G=(S,T,A,B,V,\pi)$, two isolated parties, commonly known as Alice and Bob, play against the verifier. The verifier chooses a pair of questions $(s,t)\in S\times T$ according to some probability distribution $\pi$ and sends $s$ to Alice and $t$ to Bob. The players need to respond with $a\in A$ and $b\in B$ respectively. They win if $V(a,b|s,t)=1$, where $V: A\times B\times S\times T \to \set{0,1}$ is a public verification function. The players' goal is to coordinate strategies so as to maximize their probability of winning; quantum players can use shared entanglement to improve their chances of winning. We say that a strategy is \emph{perfect} if it allows the players to win with probability one. Games which admit perfect finite-dimensional quantum but not classical strategies are known as \emph{pseudo-telepathy} games \cite{Brassard05}. It is a challenging open question to understand if the optimal success probability of a game with finite question and answer sets can always be achieved using a finite-dimensional strategy. Hence, there might exist games with entangled value $\omega^*(G)=1$ that are nevertheless not pseudo-telepathy according to our definition above (see \cite{Leung08,Regev12,MancinskaV14} for examples in related settings). In this paper we avoid addressing the this above issue by dealing only with games which admit perfect finite-dimensional strategies.
The central question of this paper is as follows:
\\[-0.1in]

\textit{
Do there exist pseudo-telepathy games that cannot be won using maximally entangled state?}\\[-0.1in]


Answering the above question in the negative would imply that maximally entangled state is sufficient for zero-error communication over a noisy classical channel. This is due to the equivalence of zero-error communication protocols and a certain type of nonlocal games outlined in \cite{ZeroError1}.

\paragraph{Previous results.} It is known that maximally entangled state is sufficient for binary\footnote{In a binary game Alice and Bob need to answer bits, \ie, $S=T=\set{0,1}$.} and unique games\footnote{In a unique game for every pair of questions $(s,t)\in S\times T$ there exists a permutation $\sigma_{st}$, such that $V(a,b|s,t)=1$ if and only if $a=\sigma(b)$.} \cite{Cleve04}. However, this is due to a trivial reason, since no entanglement is needed to win these games with probability one whenever it can be done using some entangled strategy. Therefore, these two classes of games do not contain any pseudo-telepathy games and hence are not relevant to our question. Maximally entangled state is also sufficient for the games based on graph homomorphisms \cite{David} and binary constraint system games \cite{Cleve12}. These two classes are more relevant, since they contain pseudo-telepathy games.

\paragraph{Our results.} We partially answer the above question by exhibiting a class of games for which maximally entangled state is sufficient to reach perfect performance (see \Thm{Proj}). The definition of this class is similar to that of projection games. This class subsumes both pseudo-telepathy graph homomorphism games and pseudo-telepathy binary constraint system games. In addition, we show that the answer to our central question is affirmative if and only if after the addition of consistency checks all pseudo-telepathy games remain pseudo-telepathy (see \Thm{Main}). Finally, our results show that any pseudo-telepathy weak projection game allows for a device independent certification of a maximally entangled state (see \Cor{SelfTest}). As an application, we exhibit a family of games $G_n$ that can be used to test for maximally entangled states of local dimension $\Omega(n)$ (see \Cor{Hadamard}).

Our proof technique for \Thm{Proj} is inspired by the one used in \cite{Cameron} in the context of coloring games. 

\section{Preliminaries}
Let $\M(d_A,d_B)$ denote the set of all complex $d_A\times d_B$ matrices.
We use the mapping $\vec:\M(d_A,d_B)\rightarrow \bip{d_A}{d_B}$ defined via
\[
  \vec : \ketbra{i}{j} \mapsto \ket{i}\ket{j}
\] 
for all $i\in[d_A],j\in[d_B]$ and extended by linearity. 

We now derive a formula that will often be used later. Applying the fact that $\vec(AXB\tp)=A\otimes B \vec(X)$ and $\vec(A)\ct \vec(B) = \tr(A\ct B)$, and choosing matrix $D$ such that $\vec{D}=\ket{\psi}$ we obtain
\begin{align}
\tr(A\otimes B \ketbras{\psi})  & = \bra{\psi} (A \otimes B) \ket{\psi} \nonumber \\
  &= \vec(D)\ct (A\otimes B) \vec(D) \label{eq:Formula}\\
  &= \vec(D)\ct\vec(A D B\tp) \nonumber\\
  &= \tr(D\ct A D B\tp) \nonumber
\end{align}
for all $A\in\M(d_A,d_A),B\in\M(d_B,d_B)$ and $\ket{\psi}\in\bip{d_A}{d_B}$. 

Any bipartite state $\ket{\psi}\in\bip{d_A}{d_B}$ with $d_A\le d_B$ can be expressed as
\begin{equation}
  \ket{\psi} = \sum_{i\in[d_A]} \lambda_i \ket{\alpha_{i}}\ket{\beta_i}
\label{eq:Schmidt}
\end{equation}
where $\lambda_i \ge 0$ and $\set{\ket{\alpha_i}: i\in[d_A]}$ is an orthonormal basis of $\C^{d_A}$ and the vectors $\ket{\beta_i}$ are orthonormal. This is known as a Schmidt decomposition of $\ket{\psi}$. (Note that the $\vec$ mapping shows that Schmidt decomposition is simply an alternative expression of the singular value decomposition.)
The number of nonzero $\lambda_i$ is called the Schmidt rank of $\ket{\psi}$.
We say that $\ket{\psi}$ has full Schmidt rank, if $d_A=d_B$ and $\lambda_i>0$ for all $i$. We say that a state $\ket{\psi}$ is maximally entangled if all its Schmidt coefficients $\lambda_i$ are the same.

Consider a state $\ket{\psi}\in\bip{d_A}{d_B}$ of full Schmidt rank. Then $d_A=d_B$ and $\tr_A\ketbras{\psi}=\tr_B\ketbras{\psi}$. Moreover, if we work in a Schmidt basis of $\ket{\psi}$, we have that
\[
  \vec \sqrt{\tr_B(\ketbras{\psi})} = \vec{\sqrt{\textstyle{\sum_i} \lambda_i^2 \ketbras{i}}} =  \ket{\psi}.
\]

\section{Technical lemmas}
The goal of this section is to better understand the scenario where Alice and Bob each perform a measurement on some bipartite state $\ket{\psi}$ and obtain perfectly correlated outcomes.
%
%
In the context of nonlocal games the following lemma states that maximally entangled state can be used in place of any shared entangled state whose reduced state on either party commutes with the corresponding party's measurement operators. 
\begin{lemma}
Let $\set{E_i}_{i\in[n]}, \set{F_i}_{i\in[m]}\subseteq\Pos(\bips{d})$ be measurements. Also let $\ket{\psi}\in\bips{d}$ be a state of full Schmidt rank and $D:=\sqrt{\tr_B \ketbras{\psi}}=\sqrt{\tr_A \ketbras{\psi}}$.

If $[D,E_i]=0$ for all $i\in[n]$ or $[D,F_i]=0$ for all $i\in[m]$ then for all $i\in[n], j\in[m]$ we have
\[
  \tr(E_i\otimes F_j \ketbras{\psi})=0 \Leftrightarrow 
  \tr(E_i\otimes F_j \ketbras{\Psi})=0,
\]
where $\ket{\Psi}:=\frac{1}{\sqrt{d}}\sum_i\ket{\alpha_i}\ket{\beta_i}$ and $\ket{\alpha_i}\ket{\beta_i}$ is a Schmidt basis of $\ket{\psi}$.
\label{lem:MaxEnt}
\end{lemma}

\begin{proof}
Let us work in a Schmidt basis of $\ket{\psi}$. Then $\vec D=\vec \sqrt{\tr_B \ketbras{\psi}}=\ket{\psi}$. Since $\tr(E_i\otimes F_j \ketbras{\psi})=\tr(D\ct E_i D F_j\tp)$ by \Eq{Formula} and $D\ct E_i D, F_j\tp\succeq 0$, we obtain that 
\[
  \tr(E_i\otimes F_j \ketbras{\psi})=0 \Leftrightarrow 
  D\ct E_i D F_j=0.
\]
We now assume that $[D,E_i]=0$ for all $i\in[n]$ (the other case can be proven similarly). Since $D\ct E_i D F_j\tp=D\ct D E_i F_j\tp$ and $D$ has full rank, we have 
\[
  D\ct E_i D F_j\tp=0 \Leftrightarrow 
  E_i F_j\tp = 0.
\]
Observe that  $\vec{\ket{\Psi}}=\id$ and hence $\tr(E_i\otimes F_j \ketbras{\Psi})=\tr(\id\ct E_i \id F_j\tp)=\tr(E_i F_j\tp)$. Thus we obtain 
\[
  E_i F_j\tp = 0 \Leftrightarrow \tr(E_i\otimes F_j \ketbras{\Psi})=0
\]
which completes the proof.\qed
\end{proof}

In the context of local measurements, the following lemma states that only \emph{projective} local measurements can give rise to perfectly correlated outcomes. Moreover, in such a case maximally entangled state $\ket{\Psi}$ can be used as the shared entangled state. 

\begin{lemma}
Let $\set{E_i}_{i\in[n]}, \set{F_i}_{i\in[n]}\subseteq\Pos(\bips{d})$ be two measurements and $\ket{\psi}\in\bips{d}$ be a state of full Schmidt rank and $D:=\sqrt{\tr_B \ketbras{\psi}}$. If for all \emph{distinct} $i,j\in[n]$
\begin{equation} 
  \tr(E_i\otimes F_j \ketbras{\psi})=0
\label{eq:Orthog}
\end{equation}
then for all $i\in[n]$ we have that
\begin{itemize}
  \item operators $E_i, F_i$ are projectors and
  \item $[D,E_i]=[D,F_i]=0$.
\end{itemize}
\label{lem:Commute}
\end{lemma}

\begin{proof}
First, note that $D$ can be assumed to be diagonal, if we work in the Schmidt basis of $\ket{\psi}$ and hence $D\ct=D\tp=D$. Now we rewrite Equation~(\ref{eq:Orthog}) as
\begin{equation}
  \tr(D\ct E_i D F_j\tp) = 0,
\label{eq:Orthog2}
\end{equation}
for all distinct $i, j\in[n]$. It now follows that $\tr\left(D\ct E_i D (\id-F_i\tp)\right)=0$ and hence 
\[
  \supp(D\ct E_i D)\subseteq \supp F_i\tp,
\]
where $\supp(M)$ denotes the span of the columns of $M$. Similarly, the fact that $\tr\left((\id-E_i) D F_i\tp D\ct\right)=0$ gives
\[
  \supp(DF_i\tp D\ct) \subseteq \supp(E_i).
\]
Since conjugation by full rank matrix $D$ does not change the rank, the two above inclusions imply that
\[
  \supp F_i\tp = \supp(D\ct E_i D) \;\text{ and } \;
  \supp(E_i) = \supp(DF_i\tp D\ct).
\]
Combining this with the orthogonality constraints~(\ref{eq:Orthog2}), we get that 
\[
  \supp F_i\tp \perp \supp F_j\tp \text{ and } \supp E_i \perp \supp E_j
\]
for all distinct $i,j\in[n]$. Hence, both $\set{E_i}_{i\in[n]}$ and $\set{F_i}_{i\in[n]}$ are projective measurements and moreover
\[
 E_i = \supp(DF_i\tp D\ct) \text{ and } F_i\tp = \supp(D\ct E_i D),
\]
where by slight abuse of notation we use $\supp(M)$ to denote the \emph{projector} onto the span of the columns of $M$.

We now show that $[D,E_i]=0$ for all $i\in[n]$ (the proof for $[D,F_i]=0$ is similar). From the orthogonality condition~(\ref{eq:Orthog2}) and the fact that $D\ct E_i D, F_j\tp\succeq 0$, we obtain that for all distinct $i,j\in[n]$
\[
 0 = D\ct E_i D F_j\tp = D\ct E_i D \supp(D\ct E_j D).
\]
Hence, for all distinct $i,j\in[n]$ we also have $D\ct E_i D D\ct E_j D = 0$ and thus $E_i D^2 E_j=0$, since $D$ has full rank and $D\ct=D$. Now, $E_i D^2 E_j=0$ implies that $D^2$ is block-diagonal with respect to the partition of the space corresponding to projectors $E_i$. Since in such a partition each $E_i$ is block-diagonal with the blocks being $c\id$, where $c\in\set{0,1}$, it follows that $[D^2,E_i]=0$. Hence also $[D,E_i]=0$ for all $i\in[n]$ as desired.\qed
\end{proof}

It follows from the proof of \Lem{Commute} that maximally entangled states are essentially the only states that give rise to perfectly correlated outcomes. Since the operators $E_i$ and $F_i$ commute with $D$, their off-diagonal entries corresponding to different Schmidt coefficients of $\ket{\psi}$ must be zero. Thus, the $E_i$ and $F_i$ are block-diagonal, where the blocks are labeled by distinct Schmidt coefficients of $\ket{\psi}$. Hence, the measurements $\set{E_i}$ and $\set{F_i}$ are direct sums of projective measurements each of which is performed on a maximally entangled state. 

We now establish a result similar to \Lem{Commute} for two measurements with different number of outcomes.
 
\begin{corollary}
Let $\set{E_i}_{i\in[n]}, \set{F_i}_{i\in[m]}\subseteq\Pos(\bips{d})$ be two measurements and $\ket{\psi}\in\bips{d}$ be a state of full Schmidt rank and $D:=\sqrt{\tr_B \ketbras{\psi}}$. If there exists a function $f:[n]\rightarrow[m]$ such that for all $i\in[n]$ and $j\neq f(i)$
\begin{equation} 
  \tr(E_i\otimes F_j \ketbras{\psi})=0
\end{equation}
then for all $i\in[m]$ we have that
\begin{itemize}
  \item operators $F_j, E'_j:=\sum_{i : f(i)=j}{E_i}$ are projectors and
  \item $[D,E'_j]=[D,F_j]=0$.
\end{itemize}
\label{cor:Commute}
\end{corollary}
\begin{proof}
This is exactly the statement of Lemma~\ref{lem:Commute} for measurements $\set{E'_j}_{j\in[m]}$ and $\set{F_j}_{j\in[m]}$.\qed
\end{proof}

\section{Weak projection games}

We now define a class of nonlocal games for which we will later show that maximally entangled state can be used to win with certainty, whenever it can be done using some finite-dimensional quantum strategy.
\begin{definition}
Let $G=(S,T,A,B,V,\pi)$ be a nonlocal game. We say that $G$ \emph{is weakly projective for Bob}, if for each of Bob's inputs $t\in T$ there exists an input $s\in S$ for Alice and a function $f_{st}:A \to B$ such that $V(s,t,a,b)=1$ if and only if $b=f_{st}(a)$. 

The definition for a nonlocal game that is weakly projective for Alice is similar.
We say that $G$ is a \emph{weak projection} game if it is weakly projective for Bob or for Alice.
\label{def:Proj}
\end{definition}
The term ``weak projection game'' was chosen since $G$ is called a projection game if for \emph{all} pairs $(s,t)\in S\times T$ there exists a function $f_{st}$ such that $V(s,t,a,b)=1$ if and only if $f_{st}(a)=b$. Any projection game is weakly projective for both Alice and Bob; the converse, however, does not hold.

\begin{theorem}
Suppose that a nonlocal game $G=(S,T,A,B,V,\pi)$ is weakly projective for Bob. If a shared entangled state $\ket{\psi}\in\bips{d}$ of full Schmidt rank and measurements $\mathcal{E}^{(s)}:=\set{E^{s}_i}_{i\in A}$ and $\mathcal{F}^{(t)}:=\set{F^{t}_i}_{i\in B}$ specify a perfect strategy for $G$ then
\begin{enumerate}
  \item operators $F^t_j$ are projectors for all $t\in T,j\in B$;
  \item a maximally entangled state $\ket{\Psi}$ can be used in place of $\ket{\psi}$.
\end{enumerate}
\label{thm:Proj}
\end{theorem}

\begin{proof}

To prove the theorem, for each $t\in T$ we apply Corollary~\ref{cor:Commute} to measurements $\mathcal{E}^{(s(t))}$ and $\mathcal{F}^{(t)}$, where $s(t)$ is Alice's input corresponding to $t$ from Definition~\ref{def:Proj}. This gives us item (1) and that $[F_j^t,D]=0$ for all values of $t,j$. Now, by Lemma~\ref{lem:MaxEnt} we get item (2).\qed
\end{proof}

Recalling discussion after the proof of \Lem{Commute} we immediately obtain that weak projection games allow for self-testing of maximally entangled states.
\begin{corollary}
Consider a pseudo-telepathy game $G=(S,T,A,B,V,\pi)$ which is weakly projective for Bob. Moreover, assume that any perfect strategy for $G$ must use entanglement of local dimension at least $d$. Then any  perfect strategy for $G$ certifies the presence of a maximally entangled state of local dimension at least $d$.
\label{cor:SelfTest}
\end{corollary}

To illustrate how to apply the above corollary, let us consider coloring games of Hadamard graphs $H_n$ \cite{Avis06}. The vertices of $H_n$ are all the $n$-bit strings and two vertices are adjacent if the corresponding strings differ in exactly half of the positions. In the $c$-coloring game, Alice and Bob each get asked a vertex $s$ and $t$ respectively and they need to respond with the same color when $s=t$ and different colors when $s$ is adjacent to $t$ \cite{Cleve04}. Let $G_n$ be the game for $n$-coloring the Hadamard graph $H_n$. \label{pg:Gn}
Like any of the coloring games, $G_n$ is weakly projective. It is also known that $G_n$ is a pseudo-telepathy game whenever $4|n$ and $n \ge 12$ \cite{Avis06,Godsil08}. 

\begin{restatable}{theorem}{Lower}
\label{thm:Lower}
Let $4|n$ and assume that $\mathcal{S}$ is a perfect strategy for $G_n$. Then $S$ uses entanglement of local dimension at least $\Omega(n)$.
\label{thm:EntLowerBd}
\end{restatable}

The proof of \Thm{Lower} along with a detailed discussion can be found in \Apx{Comm}. The basic idea is to translate entangled strategies for $G_n$ into one-way communication protocols for an appropriate communication version of $G_n$ (a certain promise equality problem $\EQ$). With this translation in hand, it only remains to obtain an appropriate lower on the  one-way communication complexity of $\EQ$, which we accomplish in \Lem{CommLowerBd}. The general approach used to convert lower bounds in the communication setting to lower bounds on the dimension of the shared entanglement might be of independent interest.

Combining \Thm{Lower} with \Cor{SelfTest} immediately gives the following: 

\begin{corollary}
Let $n\in \N$ be such that $4|n$ and $n \ge 12$.
Any perfect strategy for the coloring game $G_n$ certifies the presence of a maximally entangled state of dimension at least $\Omega(n)$.
\label{cor:Hadamard}
\end{corollary}



\section{Games with added consistency checks}

Given a nonlocal game $G$, we now define a new game by adding consistency check questions for one of the parties. 

\begin{definition}
Given game $G=(S,T,A,B,V,\pi)$, let $\tilde{G}_B:=(S,T,\tilde{A},B,\tilde{V},\pi)$, where 
\[ 
  \tilde{S} := \set{(s,0),(t,1) : s\in S, t\in T},
\]
$\tilde{A} := A\cup B$, and 
\[
  \tilde{V}\big(a,b|(s,i),t\big) := 
  \begin{cases}
   V(a,b|s,t)  &\text{if $i=0$}\\
   \delta_{ab} &\text{if $i=1$}\\
  \end{cases}
\] 
The game $\tilde{G}_A$ is defined similarly.
\end{definition}
It is easy to see that $\tilde{G}_A$ $(\tilde{G}_B)$ is a weak projection game for Alice (Bob) and therefore can be won using a maximally entangled state whenever some perfect quantum strategy exists. Also, any strategy used to win $\tilde{G}_A$ or $\tilde{G}_B$, can be used to win $G$.  Hence, we obtain the following:
\begin{theorem}
A nonlocal game $G$ admits a perfect finite-dimensional quantum strategy if and only if $\tilde{G}_A$ or $\tilde{G}_B$ admits some perfect finite-dimensional quantum strategy.
\label{thm:Main}
\end{theorem}

The above theorem can be used to show that a maximally entangled state is sufficient for both binary constraint system and homomorphism games. 

In a binary constraint system game $G$, Alice is asked to assign values to the binary variables in a constraint $c_s$ and Bob is asked to assign a value to a binary variable $x_t$. To win, their answers need to be consistent and Alice's assignment has to satisfy the constraint $c_s$. Any strategy used to win game $G$ can also be used to win $\tilde{G}_B$. This is because upon receiving input $(t,1)$ Alice can perform any measurement corresponding to a constraint $c_s$ that contains variable $x_t$ and respond with the value she would have assigned to the variable $x_t$.

In a homomorphism game $G$, Alice and Bob have the same input and output sets ($S=T$ and $A=B$) and their answers need to agree when they are given the same outputs. Therefore, any strategy used to win $G$ can be used to win both $\tilde{G}_A$ and $\tilde{G_B}$ (essentially by ignoring the extra bit in the modified game).

\section{Discussion}
In this paper we have looked at the question of whether there exist pseudo-telepathy games that cannot be won with certainty using a maximally entangled state in some dimension $d$. As partial progress towards answering this question we have exhibited a class of pseudo-telepathy games for which maximally entangled state is always sufficient. Additionally, we have characterized pseudo-telepathy games which admit perfect strategies with a maximally entangled state.
We hope that this characterization might help in producing an example of pseudo-telepathy which does not admit a maximally entangled perfect strategy (assuming such games exist).

An intersting open question is whether \Lem{Commute} admits an approximate version. More formally: suppose two measurements produce almost perfectly correlated outcomes, does it imply that the measurement operators almost commute with the reduced state $D$?

\paragraph{Acknowledgments.} I would like to thank Dmitry Gavinsky for suggesting the formulation of \Lem{Mix} and David Roberson for many helpful discussions. 
This research is supported by the Singapore Ministry of Education (MOE) and National Research Foundation Singapore, as well as MOE Tier 3 Grant ``Random numbers from quantum processes'' (MOE2012-T3-1-009).

\bibliographystyle{alphaurl}

\input{MaxEnt.bbl}

\newpage

\appendix
\section{Communication protocols from entangled strategies}
\label{apx:Comm}

Communication complexity was first introduced by Yao \cite{Yao79} and there have been several examples where results in the communication setting have been translated to nonlocal games. Here notable examples include Hidden Matching problem  \cite{Buhrman11} and coloring of Hadamard graphs, which can be seen as a certain promise equality problem \cite{Buhrman98}. In both cases separations between the quantum and classical communication complexity were translated into separations between the entangled and classical values of the game.

In the communication setting the goal of two separated parties, Alice and Bob, is to compute a joint function $f(s,t)$, where input $s\in S$ is given to Alice and $t\in T$ to Bob. We focus on one-way  communication protocols, where, after receiving her input $s$, Alice sends a message $m_s$ to Bob who must then correctly output $f(s,t)$. The communication cost of such a protocol is the length of the message $m_s$ (measured in bits or qubits) for the worst case input $s$. Consequently, the communication complexity of $f$ is the minimum communication cost over all valid protocols. Let us stress that throughout this section we only consider the setting where Bob must always announce the correct value of $f(s,t)$.

We are interested in the deterministic one-way communication complexity, $D^1(f)$, and its quantum analogue $Q_0^1(f)$ and in our case $f$ is a certain promise equality problem $\EQ$ (\ie, partial function). For this problem, the input sets $S$ and $T$ are the vertices of the Hadamard graph $H_n$ where $4|n$, \ie, $n$-bit binary strings. The promise is that either $s=t$ or $s$ is adjacent to $t$ (strings $s$ and $t$ differ in exactly half of the positions) and Bob must announce which case it is. In other words, $\EQ(s,t) := \delta_{st}$ for the allowed input pairs $(s,t)$. In general, if $\EQX$ is the equality problem with the promise that the parties are either given equal or adjacent vertices of a graph $X$, then  $D^1(\EQ) = \log(\chi(X))$, where $\chi(X)$ is the chromatic number of $X$ \cite{deWolf01}. It now follows that $D^1(\EQ) = \log(\chi(H_n)) = \Omega(n)$ since the chromatic number of Hadamard graphs $H_n$ is exponential in $n$ for $n=4m$~\cite{Frankl,Avis06}. To determine the quantum communication complexity $Q^1_0(\EQ)$, it is useful to consider a relaxation of coloring, called orthogonal representation. In such a representation, for each vertex $s\in V(X)$ we want to assign a unit vector $\ket{v_s}\in\C^d$ so that adjacent vertices receive orthogonal vectors. The orthogonal rank of $X$, denoted $\xi(X)$, is the smallest dimension $d$ in which $X$ admits an orthogonal representation. Similar to the classical case, it is known that \mbox{$\log(\xi(X)) = Q^1_0(\EQX)$}~\cite{Briet15}.

\begin{lemma} Let $4|n$. Then $Q_0^1(\EQ) = \log\big(\xi(H_n) \big)= \log n$.
\label{lem:CommLowerBd}
\end{lemma}
\begin{proof}
When Lov\`{a}sz introduced his Theta function in \cite{Lovasz79}, he showed that Theta of the complement, $\vartheta(\overline{X})$, lower bounds the orthogonal rank $\xi(X)$ for any graph $X$. Combining this with the fact that $\vartheta(\overline{H_n}) = n$ \cite{David}, we obtain that $\xi(H_n) \ge n$. To see that $\xi(H_n) \le n$, it suffices to note that 
\begin{equation}
  \ket{v_s} = \frac{1}{\sqrt{n}}\sum_{i\in [n]} (-1)^{s_i} \ket{i}
\end{equation}
gives an $n$-dimensional orthogonal representation of $H_n$. To conclude the proof we use the fact that $Q_0^1(\EQ) = \xi(H_n)$.\qed
\end{proof}

We summarize the one-way communication complexities for promise equality problems in \Tab{CC}.\vspace{-0.1in}

\begin{table}
\addtolength{\tabcolsep}{7pt}
\normalsize
\begin{tabular}{l|c|c}
 &  $\EQX$ & $\EQP{H_n} = \EQ$\\
\hline
One-way deterministic c.c., $D^1$ & $\log\big(\chi(X)\big)$ & $\Omega(n)$\\
One-way exact quantum c.c., $Q^1_0$ & $\log\big(\xi(X)\big)$ & $\log  n$ 
\end{tabular}\vspace{0.2in}
\caption{\normalsize One-way communication complexity of promise equality problems}
\label{tab:CC}
\end{table}\vspace{-0.3in}

According to \Lem{CommLowerBd}, any exact one-way quantum communication protocol for $\EQ$ must use messages of length $\log n$. We would like to 
transfer this lower bound to the nonlocal games setting (recall the definition of the games $G_n$ from page~\pageref{pg:Gn}).
To this end, we consider a simple construction for converting entangled strategies for $G_n$ to one-way communication protocols for $\EQ$. In fact, a similar construction can be used to relate strategies for any coloring game on a graph $X$ to one-way communication protocols for the promise equality problem $\EQX$.

\paragraph{Communication protocol for $\EQ$ from an entangled strategy for $G_n$.}\label{Construction}
Consider a perfect entangled strategy $\mathcal{S}$ for the nonlocal game $G_n$. Let $\ket{\psi}\in \C^d\otimes \C^d$ be the shared entangled state used in $\mathcal{S}$. Also let $\mathcal{A}^s$ and $\mathcal{B}^t$ be the measurements they apply upon receiving $s$ and $t$ respectively. In the communication setting, upon receiving input $s$, Alice prepares the bipartite state $\ket{\psi}$ and measures its first register using $\mathcal{A}^s$. She then sends the measurement outcome $a\in[n]$ along with the second register of the state to Bob. This requires communication of $\log n $ classical bits and $\log d$ qubits. After receiving the quantum state, Bob measures it using $\mathcal{B}^t$ and compares the outcome of his measurement, $b$, to Alice's outcome $a$ that she had sent over. He announces that $s=t$, if $a=b$ and otherwise he says that $s\neq t$. The correctness of Bob's answer follows directly from the fact that the strategy $\mathcal{S}$ is perfect for $G_n$.

\medskip

We now show that quantum one-way communication protocols for $\EQ$ in which Alice also sends Bob a small amount of classical information do not allow for significant savings in the quantum communication cost. In combination with the above translation of entangled strategies into communication protocols this will allow us to lower bound the dimension of the entanglement needed for a perfect strategy of $G_n$.

\begin{lemma}
Let $\mathcal{P}$ be an exact one-way communication protocol for the promise equality problem $\EQ$ which uses $\log n$ classical bits and $\log d$ qubits. Then \mbox{$d = \Omega(n)$}, \ie, the quantum communication cost of $\mathcal{P}$ is $\Omega(\log n)$. 
\label{lem:Mix}
\end{lemma}

\begin{proof}
Without loss of generality, we can assume that Alice's actions in the protocol $\mathcal{P}$ are deterministic. That is, upon receiving input $s\in \set{0,1}^n$ she transmits a classical message $m_s$ of at most $\log n$ bits and a quantum state $\ket{\psi_s}\in\C^d$. The messages $m_s$ give rise to a partition of the vertex set of the Hadamard graph $H_n$ into induced subgraphs $\set[\big]{L_m: m\in \set{0,1}^{\log n}\cong[n]}$. Therefore, $\mathcal{P}$ can be used to obtain one-way quantum communication protocols of cost $\log d$ for each of the promise equality problems $\EQP{L_m}$. Therefore, 
\begin{equation}
  \log \big(\xi(L_m)) = Q_0^1\big(\EQP{L_m}\big)\le \log d.
\end{equation}
Since $\chi(X)\le(1+2\sqrt{2})^{2\xi(X)}$ for any graph $X$~\cite{Cameron}, we obtain that 
\begin{equation}
  \chi(L_m) \le (1+2\sqrt{2})^{2d}<14^d.
\end{equation}
Therefore, $D\big(\EQP{L_m}\big)\le \log(14^d)<3d$ for all $m\in[n]$. Together with the classical messages $m_s$ from the protocol $\mathcal{P}$ this yields a deterministic one-way communication protocol for $\EQ$ of cost $\log n + 3d$. By \Lem{CommLowerBd}, we get that $\log n+ 3d =\Omega(n)$ and hence $d=\Omega(n)$ as desired.\qed
\end{proof}

We are now ready to establish a lower bound on the amount of entanglement required in any perfect strategy for $n$-coloring the Hadamard graph $H_n$.

\Lower*

\begin{proof}
Consider any perfect strategy $\mathcal{S}$ for $G_n$ and let $d$ be the local dimension of the shared entangled state employed in $\mathcal{S}$. The above construction (see page~\pageref{Construction}) allows us to convert $\mathcal{S}$ into an exact one-way communication protocol for $\EQ$ which uses $\log n$ bits and $\log d$ qubits. By \Lem{Mix} we get that $d=\Omega(n)$ as desired.\qed
\end{proof}

\end{document}

%% file: MaxEnt.bbl
\newcommand{\etalchar}[1]{$^{#1}$}